\documentclass[12pt]{article}
\usepackage{cgw03}
\usepackage{amsmath}
\usepackage{latexsym}           
\usepackage{amssymb}
\usepackage{amsmath}
\usepackage{amsthm}
\usepackage{here}
\usepackage{verbatim}
\usepackage{graphicx}
\usepackage[normal]{subfigure}
\usepackage{color}
\usepackage{xcolor}
\usepackage{lscape}
\usepackage{rotating}
\usepackage{float}
\usepackage[breaklinks]{hyperref}
\usepackage[authoryear]{natbib}
\bibpunct{(}{)}{;}{a}{,}{,} 
\definecolor{DarkBlue}{rgb}{0.1,0.1,0.5}
\definecolor{Red}{rgb}{0.9,0.0,0.1}
\definecolor{Yellow}{rgb}{1,1,0}
\definecolor{Orange}{rgb}{1,0.5,0.25}
\definecolor{DarkOrange}{rgb}{1,0.5,0}
\definecolor{Violent}{rgb}{0.5,0,1}
\definecolor{DarkPink}{rgb}{1,0,1}
\definecolor{DarkGreen}{rgb}{0.5,0.8,0.5}
\newtheorem{definition}{Definition}[section]
\newtheorem{lemma}{Lemma}[section]
\newtheorem{theorem}{Theorem}[section]

\newtheorem{assumption}{Assumption}[section]
\theoremstyle{definition}
\newtheorem{example}{Example}[section]

\definecolor{szary1}{gray}{0.85}
\definecolor{szary2}{gray}{0.5}

\renewcommand{\Re}{\text{Re}}

\newcommand{\ciagn}[1]{\{#1 : t\in \mathbb{N}\}}

\newcommand{\cosec}{\text{cosec}}

\newcommand{\ciag}[1]{\{#1 : t\in \mathbb{Z}\}}

\newcommand{\ma}{$2\!\times\!12\text{MA }$}

\newcount\footnotenumber
\def\numberedfootnote{%
  \global\advance\footnotenumber by 1
  \footnote{$^{1}$}
  }

\newtheoremstyle{Satz}
  {1cm}                   
  {1cm}                   
  {\sffamily}           
  {}                      
  {\normalfont\bfseries}  
  {:}                     
  {\newline}              
  {\underline{\thmname{#1}\thmnumber{ #2}\thmnote{ (#3)}}}
\theoremstyle{Satz}

\title{Almost Periodically Correlated Time Series in Business Fluctuations Analysis\footnote{$\quad$ Paper presented at FENS 2012, Symposium on Econo- and Sociophysics, 19-21.04.2012, University of Gdansk and Gdansk University of Technology,  Poland}}
\author {\L{}ukasz Lenart\thanks{$\quad$ Economic Institute in National Bank of Poland, Department of Mathematics in Cracow University of Economics, e-mail: lukasz.lenart@uek.krakow.pl}, Mateusz Pipie\'{n}\thanks{$\quad$ Economic Institute in National Bank of Poland, Department of Econometrics and Operations Research in Cracow University of Economics, e-mail: eepipien@cyf-kr.edu.pl}}

\begin{document}
\maketitle

\begin{abstract}
We propose a non-standard subsampling procedure to make formal statistical inference about the business cycle, one of the most important unobserved feature characterising fluctuations of economic growth. We show that some characteristics of business cycle can be modelled in a non-parametric way by discrete spectrum of the Almost Periodically Correlated (APC) time series. On the basis of estimated characteristics of this spectrum business cycle is extracted by filtering. As an illustration we characterise the man properties of business cycles in industrial production index for Polish economy.
\\ \textbf{Keywords:}  business cycle, industrial production index, almost periodically
correlated time series, subsampling procedure.
\end{abstract}
\section{Introduction}
Seminal works, that originated interest in empirical modelling of business cycles in macroeconomy, clearly indicated theirs inhomogeneity for both, spatial and time domain. In particular, introductory remarks of W.C. Mitchell in \emph{Business Annals}, \cite{annals}, contains the following suggestion: \emph{No two recurrences in all the array seem precisely alike. Business cycles differ in their duration as wholes and in the quickness and the uniformity with which they sweep from one country to another}. When identifying business cycle R.E. Lucas proposed its own conception, which, as he states in his 1977 paper, identifies the business cycle with \emph{movements about trend in gross national income}. These movements are typically irregular in period and in amplitude. Regularities are only observed \emph{in the co-movements among different aggregative time series}; see \cite{lucas1977}. It is clear, that Mitchell initially suggested different time pattern of business cycles for different economies. However, it is obvious, that from the dynamic point of view, as Lucas states, business cycle exhibit irregular and nonperiodic character.
\\For developed economies some stylised facts about business cycles are known in the literature; see \cite{king_rebelo} or \cite{Stock_Watson}. But, we see the lack of precise and well established methods of formal statistical modelling of those empirical properties. It prompts new studies resulting many different approaches and frameworks of business cycle extraction; see for example exhaustive review presented \cite{Diebold_Rudenbush}. When the lack of the theory of statistical inference seems to be a persistent state, the consensus about empirical properties of business cycles is based, either on an ad-hoc reasoning, or on the empirical results, that are possible to confirm using a group of methods, built on the basis of relative different frameworks. However, the extraction of the business cycle component from observed time series is still a controversial issue. In particular, since there is ongoing interest in many approaches to separate growth component from the cyclical component, and because there is no consensus on how to detrend the data, the business cycle stylised facts are sensitive to the adopted procedure. Hence, this has become not only a controversial issue in the business cycle theory itself, but also a subject of criticism by competing empirical approaches, as well.
\\The main purpose of this paper is to present a novel approach to formal business cycle estimation. We propose a non-standard subsampling procedure, in order to make formal statistical inference about the properties of the business cycle. We show, that business cycle can be modelled by parameters of discrete spectra of the Almost Periodically Correlated (APC) stochastic process. The APC class is a generalisation of Periodically Correlated (PC) class of stochastic processes, introduced by \citet{gladyshev1}. The vast literature confirmed substantial empirical importance and flexibility of PC class in many time series applications, see: \cite{parzen}, \cite{osborn}, \cite{franses_96}, \cite{franses_96_book}, \cite{Bollerslev_96}, \cite{franses_97}, \cite{Burridge_01}, \cite{franses_05}. According to \cite{hurd_miamee}, the periodically correlated stochastic processes are nonstationary, but non-constant unconditional expectation of the process exhibit periodic, and hence regular, evolution in time domain. The generalisation presented in this paper assumes that the mean of the nonstationary time series can be described by almost periodic function, i.e. the function, that belongs to the topological closure of periodic class of functions.
\\From the definition, APC stochastic processes may describe irregular character of unconditional means for nonstationary time series. Assuming, that detrended time series follows APC, we relax assumption of stationarity of cyclical factor, very commonly imposed in filtering approaches. Nonstationarity of the cycle component of the series, together with possible irregularities in time pattern of the unconditional mean, makes our approach relatively flexible and general. Consequently, incorporating the APC factor into the model of observed discrete time series should result in much more accurate approach to business cycle extraction than those proposed so far.
\\In the empirical part of the paper we analyse the cyclical behaviour of production sector in Poland with the use of the model with APC stochastic component. We characterise business cycle on the basis of industrial production index and also on some subsector indices. We discuss the empirical results and reflect them to the previous analyses conducted for the Polish economy.
\section{Basic notation and definitions}
Formally, the second-order and real-valued time series $\{X_{t}: t\in
\mathbb{Z}\}$ is called periodically correlated if both the mean function $\mu(t)=E(X_{t})$ and the autocovariance function
$B(t,\tau)=\text{cov}(X_{t},X_{t+\tau})$ are periodic at $t$ for every $\tau \in
\mathbb{Z}$, with the same period $T$. In order to introduce the class of almost periodically
correlated time series we need the definition of almost periodic function. We recall
the following definition from \cite{Cor}:
\begin{definition}
A real-valued function $f(t): \mathbb{Z} \longrightarrow  \mathbb{R} $ of an integer
variable is called almost periodic (AP in short), if for any $\epsilon>0$ there
exists an integer $L_{ \epsilon}>0$, such that among any $L_{ \epsilon}$ consecutive
integers, there is an integer  $p_{\epsilon}$ with the property
$$\sup_{t \in \mathbb{Z}}|f(t+p_{ \epsilon})-f(t)|< \epsilon.$$
\end{definition}
A second-order real-valued time series $\ciag{X_{t}}$ is called almost
periodically correlated if both the mean function $\mu(t)=E(X_{t})$ and the autocovariance function $B(t,\tau)=\text{cov}(X_{t},X_{t+\tau})$ are
almost periodic function of an integer variable, for every $\tau \in \mathbb{Z}$. It
is easy to see that any periodic function is almost periodic. Therefore, the class of
APC time series is wider than the class of PC time series. During last five decades the APC class was broadly applied in telecommunication (\cite{gardner86}, \cite{Nap01}), climatology \cite{Blo_Hurd_Lund_94}) and many others fields. For exhaustive review of possible applications see \cite{gardner2} and \cite{giannakis_05}). Empirical importance of such a class of nonstationary time series prompted new studies concerning properties and estimation methods.
\\In APC case the mean function and the autocovariance function $B(t,\tau)$ for any $\tau \in \mathbb{Z}$ has the Fourier representation of the form:
\begin{equation}
\mu (t) \sim \sum\limits_{\psi \in \Psi } m(\psi)e^{i \psi t},\:\;\:\:\:
B(t,\tau) \sim \sum\limits_{\lambda \in \Lambda_{\tau}}
a(\lambda,\tau)e^{i \lambda t}, \label{rep1}
\end{equation}
where the Fourier coefficients  $m(\psi)$ and $a(\lambda,\tau)$  are given by:
\begin{equation}
m(\psi)=\lim\limits_{n \rightarrow \infty}\frac{1}{n}\sum\limits_{t=1}^{n}\mu(t)e^{-i \psi t},\:\;\;\;
a(\lambda,\tau)=\lim\limits_{n \rightarrow \infty}\frac{1}{n}\sum\limits_{j=1}^{n}B(j,\tau)e^{-i \lambda j},
\end{equation}
see \cite{Hur_cor_91}, \cite{hurd_miamee}). According to \cite{Cor} sets $\Psi=\{\psi \in [0, 2\pi): m_{X}(\psi)\not =0 \}$ and $\Lambda_{\tau}=
\{\lambda \in [0,2\pi): a(\lambda,\tau)\not = 0 \}$ are countable.
Hence, the set $\Lambda = \bigcup_{\tau \in \mathbb{Z}} \Lambda_{\tau}$ is also
countable. If the time series is PC, then representations
(\ref{rep1}) become equations and the sets $\Psi$ and $\Lambda$
are contained in the set $\{ 2 k\pi/T:k=0,1,\ldots,T-1 \}$.
\\In the problem of business cycles extraction the vast econometric literature exploit approaches based on the assumption of zero mean imposed on the distribution of stochastic factor describing business fluctuations. Moreover, this stochastic factor is usually modelled under stationarity assumption, leading to the framework that utilises parameters of continuous spectrum. The econometric approach presented in this paper relaxes stationary assumption, and consequently a more general dynamic model of observed time series is subject to empirical analysis. We model business cycles in a non-parametric way, taking into account discrete spectra of observed time series.  It means that we characterise business cycles by non-zero frequencies $\psi \in \Psi$ and by corresponding Fourier coefficients $m(\psi)$. The definition and properties of discrete spectra in
simple representation see \cite{priestley}, or in PC case in \cite{hurd_miamee}.
\\Notice that any $\psi_0 \in (0,2\pi)$ corresponds to the length of the cycle $2 \pi / \psi_0$. Hence the following testing problem:
\begin{equation}
\begin{array}{ll}
H_0: &  \psi_0 \not \in \Psi\\
H_1: & \psi_0 \in \Psi,
\end{array}
\label{hipotezaa}
\end{equation}
enables to test the statistical significance of the cycle with appropriate length. According to the definition of the set $\Psi$ and Fourier coefficients $m(\cdot)$ our testing problem is equivalent to the following:
\begin{equation}
\begin{array}{ll}
H_0: &  |m(\psi_0)| =0 \\
H_1: & |m(\psi_0)| \not = 0.
\end{array}
\label{hipotezab}
\end{equation}
We consider formulation \ref{hipotezab} in details. Since we are interested in business cycle estimation we restrict frequency $\psi_0$ such that corresponding length of the cycle is not shorter then 1.5 years. This formally means, that in further analysis for monthly data we assume that $\psi_0 \in (0,0.35)$.
\\In this paper by amplitude, which corresponds to frequency $\psi \in \Psi  \cap (0,0.35)$, we define the distance between maximum and minimum value of the function $h(t)=2\Re[m(\psi) e^{i \psi t}]$.
\\The problem stated above requires statistical theory of detecting significant frequencies in the set $\Psi$ and corresponding Fourier coefficients $m(\cdot)$. In the next sections we present some results given sampling model generated by APC assumption.
\section{Estimation problem}
By $\{ X_{c_n+1},X_{c_n+2}, \ldots, X_{c_n+d_n}\}$ we denote a sample from APC time series $\ciag{X_{t}}$, where  $\{d_n\}_{n \in \mathbb{N}}$ is any sequence of integers tending to infinity with $n$ and  $\{c_{n}\}_{n \in \mathbb{N}}$ is any sequence of integers. For any $\psi \in [0,2\pi)$, estimator $\hat
m_{n}^{c_n,d_n}(\psi)$ ($\hat m_{n}^{c,d}(\psi)$ for short) of the parameter
$m(\psi)$, in representation (\ref{rep1}), based on sample  $\{ X_{c_n+1},X_{c_n+2},
\ldots, X_{c_n+d_n}\}$ takes the form:
\begin{equation}
\hat m_{n}^{c,d}(\psi) = \frac{1}{d_n}\sum\limits_{j=c_n+1}^{c_n+d_n}X_{j}e^{-i \psi j}.
\label{miu}
\end{equation}
The standardised version of (\ref{miu}) has asymptotic normal distribution with zero mean; see \cite{Lenart_11_jtsa}, Theorem 2.1. Additionally, the variance of this distribution is a function of values of the generalised spectral density, calculated at arguments dependent on $\psi$. For definition and basic properties of generalized spectral density in APC case see \cite{hurd_rep}, \cite{Deh_hurd_93}). Since the standard theory, presented in \cite{Lenart_11_bernoulli}, provides methods of estimation of generalised spectral density in APC case only under the zero-mean assumption or under the assumption that the set $\Psi$ is known and finite, estimation of generalised spectral density in our case is not possible so far. Therefore, in the paper we exploit subsampling methodology, to construct asymptotically consistent test related to (\ref{hipotezab}). In this approach the asymptotic variance estimation is not of particular interest. Similarly, subsampling methodology was also used for PC case in time domain in \cite{Len08b} and for APC case in frequency domain in \cite{Lenart_11_bernoulli}.
\\The problem of frequency estimation $\psi_0$ can be solved on the basis of a more generalised approach than presented by \cite{walker_71}. Given assumption that there exists interval $I_{\psi_0}$, such that $I_{\psi_0} \cap \Psi = \{\psi_0\}$, it is possible to formulate the natural estimator of the unknown frequency $\psi_0$ of the form $\hat \psi_{n} = \arg \, \max_{x \in I_{\psi_0}}\{ \sqrt{n}|\hat m_{n}(x)|\}.$ As it was shown in \cite{Lenart_11_jtsa}, Theorem 3.1, under some regularity conditions we have:
\begin{equation}
\ \left[
\begin{array}{c}
 \hat m_{n}^{c,d}(\hat \psi_{n})\\
 \hat \psi_{n}
  \end{array}
\right] \stackrel{p}{\longrightarrow } \left[
\begin{array}{c}
m(\psi_0)\\
\psi_0
\end{array}
\right].
\label{eqww}
\end{equation}
\section{Subsampling procedure and consistency}
In this section  we describe the main idea of the subsampling methodology, according to the approach presented and developed by \cite{Pol}. We use the same notation. Initially we assume that the time series $\ciag{X_{t}}$ is governed by unknown probability distribution $P$, that belongs to a certain class of probability measures $\mathcal{P}$. Denote by $\{X_1,X_2,\ldots,X_n\}$ a sample from the time series $\ciag{X_{t}}$. Our goal is to approximate the
distribution of
\begin{equation}
\upsilon_{n}( \hat \theta_{n} - \theta(P)),
\label{root}
\end{equation}
where $\hat \theta_{n}=\hat \theta_{n}(X_1,X_2,\ldots,X_n)$ is an estimator of
$\theta(P)$, the parameter of interests, and $\upsilon_{n}$ is appropriate
normalising sequence.  Let $b(n)$  ($b$ for short) be any sequence of integer numbers
tending to infinity with $n$, such that $b<n$ and $b/n \to 0$.
\\One of the main assumption in subsampling methodology is that there exists asymptotic distribution of (\ref{root}). We denote this distribution by  $J(P)$, with $J(x,P)$ as a corresponding cumulative distribution functions at point $x \in \mathbb{R}$. Following the idea of \cite{Pol} the distribution of (\ref{root}) can be approximated by its subsampling version of the form:
\begin{equation}
L_{n,b}(x)=\frac{1}{n-b+1}\sum\limits_{t=1}^{n-b+1}\mathbf{1} \{ \upsilon_{b}(
\hat \theta_{n,b,t} - \hat \theta_{n}) \leq x \},
\label{esty_podprobkowy}
\end{equation}
where $\mathbf{1}{\{B\}}$ is the indicator function of the set $B$ and $\hat{\theta}_{n,b,t}=\hat{\theta}_b(X_t,X_{t+1},\ldots,X_{t+b-1})$ as an estimator of the unknown parameter $\theta(P)$ obtained on the basis of the sample $\{X_t,X_{t+1},\ldots,X_{t+b-1}\}$, with $t$ as a starting point and $b$ as a size of subsample.
Under suitable regularity conditions  it is known that (see \cite{Pol}, Theorem 4.2.1, page 103):
\begin{itemize}
\item [(i)]if $x$ is a continuity point of $J(\cdot,P)$, then $L_{n,b}(x) \stackrel{p}{\longrightarrow } J(x,P),$
\item [(ii)]if $J(\cdot,P)$ is continuous, then $\sup_{x \in
    \mathbb{R}}|L_{n,b}(x) - J(x,P)|\stackrel{p}{\longrightarrow }0,$
\item [(iii)]If $J(\cdot,P)$ is continuous at point $c(1-\alpha)$, then
\begin{equation}
P\left(\upsilon_{n}(\hat \theta_{n}-\theta(P))\leq
c_{n,b}(1-\alpha)\right) \to 1-\alpha,
\end{equation}
where for any $\alpha \in (0,1)$, we define
\begin{equation}
c_{n,b}(1-\alpha) =\inf \{x:\,L_{n,b}(x) \geq  1-\alpha\},
\nonumber
\end{equation}
\begin{equation}
c(1-\alpha) =\inf \{x:\,J(x,P) \geq  1-\alpha \}.
\nonumber
\end{equation}
\end{itemize}
The implication (iii) is crucial to construct a subsampling confidence interval for the parameter $\theta(P)$.
\\We are interested in estimation of the absolute value of coefficients of the Fourier representation of the mean of APC process. Namely, we take  $\theta(P)=|m(\psi)|$. Subsampling procedure, with $\hat \theta_{n,b,t}=|m_{n}^{t-1,b}(\psi)|$ and $\upsilon_{n}=\sqrt{n}$, is consistent; see for details \cite{Lenart_11_jtsa}, Theorem 2.3. Consequently, the confidence intervals for the parameter $\theta(P)=|m(\psi)|$, obtained by subsampling procedure, are asymptotically consistent.
\\Now let take any $\psi_0 \in (0,\pi]$. The test (\ref{hipotezab}) with test statistics $\Pi_{n}(\{\psi\})=\sqrt{n}| \hat m_{n}(\psi)|$ and subsampling critical value are asymptotically consistent.
In our paper we prove some modification of this result (see Theorem \ref{subsamplingp-mean} in Appendix).
We use test statistics $\tilde \Pi_{n}(\{\psi\}) = \sqrt{n}| \hat r_{n}(\psi)|$, that can be interpreted as a value of test statistics $\Pi_{n}(\{\psi\})=\sqrt{n}| \hat m_{n}(\psi)|$ based on the sample $\{X_{1}-\overline{\textbf{X}}_{n},X_{2}-\overline{\textbf{X}}_{n},
\ldots,X_{n}-\overline{\textbf{X}}_{n} \}$, where $\overline{\textbf{X}}_{n}$ is the sample mean for the path $\{X_{1},X_{2},\ldots,X_{n}\}$. The critical value $\tilde c_{n,b}^{\{\psi\}}(1-\alpha)$ is calculated according to the formula that utilises subsampling procedure:
$$\tilde c_{n,b}^{\{\psi\}}(1-\alpha)=\inf\{x: \tilde
L_{n,b}^{\{\psi\}}(x) \geq 1-\alpha \},$$ $$ \tilde
L_{n,b}^{\{\psi\}}(x)=\frac{1}{n-b+1}\sum_{t=1}^{n-b+1} \mathbf{1}\{
\sqrt{b}(|\hat r^{t-1,b}_{n}(\psi)| - |\hat r_{n}(\psi)| )\leq x \},$$
where
\begin{equation}
\hat r_{n}^{c,d}(\psi) = \frac{1}{d_n}\sum\limits_{j=c_n+1}^{c_n+d_n}(X_{j}-\overline{\textbf{X}}_{n})e^{-i \psi j}.
\nonumber
\end{equation}
and $\hat r_{n}(\psi)=\hat r_{n}^{0,n}(\psi)$.
\section{Statistical model of cyclical fluctuations}
In this section we present the statistical framework of extraction the cyclical component when the one-dimensional time series describing changes in economic activity is observed. In Section \ref{par_model} we present basic assumptions concerning the model, while in section \ref{par_estymacja} we describe in details the algorithm of formal statistical extraction of business cycle component.
\subsection{Model structure and assumptions}\label{par_model}
Let consider a real-valued time series, denoted by $\ciag{P_{t}}$. At the beginning of this section we assume that the unconditional expectation for the process $\ciag{P_{t}}$  exists for any $t \in \mathbb{Z}$.
\\An interesting case, that is of particular interest in econometrics is the class of integrated stochastic processes, denoted by $I(d)$ for integration of order $d\in \mathbb{N}$. If we are interested in analysis of $I(d)$ processes in our framework, some additional assumptions should be imposed top assume the existence of unconditional moments. In the case when $\ciag{P_{t}}$ is $I(1)$ process it is sufficient to assume additionally that there exists $t_0 \in \mathbb{Z}$ such that $E(P_{t_{0}})<\infty$. Hence we obtain in this case, that $E(P_{t})< \infty$  for any $t \in \mathbb{Z}$. More generally, if $\ciag{P_{t}}$ follows $I(d)$ process, then it is sufficient to assume that there exists $t_{0} \in
\mathbb{Z}$, such that $E(P_{t_0 + k})<\infty$  for $k=0,1,\ldots,d-1$, to assure moment existence. Consequently, we formally exclude in our analysis processes with pure integration, but some restricted cases, representing strict nonstationarity with finite unconditional mean, may be modelled.
\\For further analysis we assume that the mean function $\mu_{P}(t)=E(P_{t})$ is defined by the sum of deterministic function $f(t,\beta)$, parameterized by $\beta \in \mathbb{R}^p$, and almost periodic function $g(t)$, with the Fourier expansion of the form:
\begin{equation}
g(t)=\sum\limits_{\psi \in \Psi_{P}}m_{P}(\psi)e^{i \psi t}.
\label{fouriera_produkcja1}
\end{equation}
For convenience, we rewrite $g(t)$ in equivalent representation:
$$g(t)=\sum_{\psi \in \Psi_{P} \cap [0,\pi] }a_{P}(\psi)\cos( \psi t) + b_{P}(\psi)\sin( \psi t).$$
This automatically implies, that:
\begin{equation}
\mu_{P}(t)=f(t,\beta) + g(t) = f(t,\beta) +\sum\limits_{\psi \in \Psi_{P}}m_{P}(\psi)e^{i \psi t}.
\label{fouriera_produkcja}
\end{equation}
Equation (\ref{fouriera_produkcja}) leads to a more general approach to modelling business fluctuations, than those presented in the literature so far; see for example: \cite{Beveridge_81}, \cite{Clark_87}, \cite{harvey_93}, \cite{Hamilton_89}, \cite{krolzig_97}. The main advantage of our approach is, that it only relies on the specification of the first moment of the time series $\ciag{P_{t}}$, making model assumptions much weaker. To illustrate the importance of our assumptions and generalisation we present an example below.
\begin{example}
Let $\ciag{P_t}$ be a time series such that $P_t=P_{t-1}+\epsilon_{t},$ where  $E(P_{0})=b$ and $\ciag{\epsilon_{t}}$ is APC time series with expectation function  $\mu_{\epsilon}(\cdot)$ such that $\mu_{\epsilon}(t)=a+g(t)-g(t-1),$
where $g: \mathbb{Z} \to \mathbb{R}$ is a function of the form:
$$g(t)=\sum\limits_{\psi \in \Psi} m(\psi)e^{i \psi t},$$
$a \in \mathbb{R}$ and $\text{card}(\Psi) < \infty$.
Notice that for any $t \geq 1$ we have
$$P_t=P_{0}+\epsilon_{1}+\epsilon_{2}+\ldots+\epsilon_{t}.$$
Therefore
$$E(P_t)=b+\sum\limits_{j=1}^{t}E(\epsilon_{j})= b+ a t -g(0) + g(t) = f(t,\beta) + g(t),$$
where $f(t,\beta)=\beta_0 + \beta_1 t$, $\beta_0=b-g(0)$, $\beta_1=a$. This means that time series  $\ciag{P_t}$ can be represented as  (\ref{fouriera_produkcja}). If $g(t)\equiv 0$, then $\mu_{\epsilon}(t)=a$, and time series $\ciag{P_{t}}$ can be interpreted as $I(1)$ process with drift and assumption $E(P_{0})=b$.
\end{example}
The function $f(t,\beta)$ can be interpreted as a trend component, modelled in this paper by the polynomial. The function $g(t)$ contains summarised information about seasonal fluctuations, business fluctuations and long-term cyclical fluctuations. From the Fourier representation of $g(t)$  we split the whole set $\Psi_{P}$ of non-zero frequencies into the mutually exclusive sets, that are related to those three periodic attributes of time series dynamics. Initially, we interpret long-term cyclical fluctuations as those with the length more than $8$ years, since the frequency $\omega$ is related to the length of cycle that equals $2\pi/\omega$ units. In order to distinguish cyclical fluctuations from seasonal fluctuations we assume formally, that in the representation (\ref{fouriera_produkcja1}) the set $\Psi_P=\{\psi:m(\psi)\not =0\} \subset [0,2\pi)$, is unknown. For the set $\Psi_P$, let consider the following decomposition:
\begin{equation} \Psi_{P}=\Psi_{P,1} \cup \Psi_{P,2} \cup \Psi_{P,3}.
\end{equation}
We assume that  $\Psi_{P,1} \cap (0,0.35) =\Psi_{P,1}$, and consequently the set $\Psi_{P,1}$ represent all frequencies with corresponding length of the cycle greater than 17 months. Therefore the set $\Psi_{P,1}$ contains frequencies that can describe business fluctuations. The set $\Psi_{P,2}$ contains only seasonal frequencies, namely $\Psi_{P,2} \subset \{2 k \pi/12:\, k=0,1,\ldots,11\}$ while $\Psi_{P,3}$ contains all remained frequencies.
In the following section we concentrate our attention only to parameter identification and estimation in the set $\Psi_{P,1}$.
\subsection{Cycle identification and estimation}\label{par_estymacja}
Our approach aims at identification and estimation of cyclical fluctuations. In order to remove trend component and to weaken seasonal effects, observed time series is subject to some preliminary transformations, . Hence, we formulate the algorithm of frequency identification that consists of three basic steps. The first step enables to remove seasonal component, the second step detects the trend component, while in the third step, parameter identification and estimation is provided.
\begin{itemize}
\item[\textbf{Step 1 - }] \textbf{removing the seasonal component.} Seasonality appear in most monthly economic time series. More formally, we allow (it is assumed), that $\Psi_{P} \cap \{
2 k \pi/12:\;k=1,2,\ldots,11 \} \not =  \emptyset.$ Since, the estimation of the frequencies and corresponding
Fourier coefficients, that represent seasonal frequencies is not of particular importance in our paper, we use centered moving average filter \ma (see:
\cite{makridakis}, \cite{Brocwell}) to remove seasonal pattern. We show below that this filter does not change the elements of the set $\Psi_{P,1}$, what is crucial for future estimation procedure. Denote by $\ciag{Y_{t}}$ time series obtained by application of the centered moving average filter. It means, that $Y_{t}=L_{2 \times 12}(B)P_{t},$ where
$$L_{2 \times 12}(B) = (B^{-6}+2B^{-5}+\ldots+2B^{-1}+2+2B+\ldots+2B^{5}+B^{6})/24,$$
and $B^k P_{t}=P_{t-k}$ for any $t$ and $k$. Note that the expectation of the
time series $\ciag{Y_{t}}$ exists. On the basis of the Theorem \ref{tw-filtracja} and elementary calculations we get
\begin{equation}
\mu_{Y}(t)=E(Y_{t})=\underbrace{\tilde{\beta}_{0} + \tilde{\beta}_{1} t  + \ldots + \tilde{\beta}_{p} t^{p}}_{\tilde f(t,\tilde \beta)} + \sum\limits_{\psi \in \Psi_{Y}} m_{Y}(\psi) e^{i \psi t},
\label{mu_Y}
\end{equation}
\\where $\Psi_{Y} \cap \{ 2 k \pi/12:\;k=1,2,\ldots,11 \} =  \emptyset$,
$\Psi_{Y} = \Psi_{P} \setminus \{ 2 k \pi/12:\;k=1,2,\ldots,11 \}$ and $\tilde{f}$ is a function. Fourier
coefficients $ m_{P}(\psi)$ and $m_{Y}(\psi)$ are related according to the formula:
\begin{equation}
m_{Y}(\psi)=L_{2 \times 12}(e^{-i \psi})m_{P}(\psi).
\end{equation}
Notice that $\tilde f(t,\tilde \beta)$ is also a polynomial of order $p$. In particular, for $s \in \{p-1,p\}$ we have $\tilde{\beta}_{s}$ = $\beta_{s}$. Consequently, given model with $p=0$ or $p=1$ (i.e. constant or linear trend) we have, that $\tilde f(t,\tilde \beta) \equiv f(t, \beta)$. In case $p=2$ functions $\tilde{f}$ and $f$ have different values, but $\tilde f(t,\tilde \beta) - f(t, \beta)$ is constant over time. Additionally, filtering the series with centered moving average operator, we obtain, that:
$$ \Psi_Y \cap (0,0.35) = \Psi_{P,1}$$
and
$$\Psi_Y \cap \Psi_{P,2}=\emptyset,$$
which means that the set $\Psi_Y$ still contains the same elements as $\Psi_{P,1}$ and dos not contain seasonal frequencies from the set  $\Psi_{P,2}$.
\item[\textbf{Step 2 -}] \textbf{removing the trend component.} The case when $p=0$ is trivial. Let consider the case $p=1$. Application of difference operator $L_{1}(B)=(1-B)$ for the time series $\ciag{Y_{t}}$ results with time series $\ciag{X_{t}}$:
$$X_{t}=L_{1}(B)Y_{t}=Y_{t}-Y_{t-1}=(P_{t+6}-P_{t-5} + P_{t+5}-P_{t-6})/24.$$
The expectation of the time series  $\ciag{X_{t}}$ exists and is described by almost
periodic function of the form:
\begin{equation}
\mu_{X}(t)= {\beta}_{1} + \sum\limits_{\psi \in \Psi_{X}}m_{X}(\psi)e^{i \psi t},
\label{fouriera_produkcja12}
\end{equation}
where \begin{equation} \Psi_{X} \subset \{ 0\}  \cup \Psi_{P}  \setminus \{ 2 k
\pi/12:\;k=1,2,\ldots,11 \}, \label{for0}
 \end{equation}
and
\begin{equation}  \Psi_X \cap (0,0.35)= \Psi_{P,1}, \label{for00}
\end{equation}
which follows from the Theorem \ref{tw-filtracja}. Additionally, we have:
\begin{equation}
m_{X}(\psi) = L_{1}(e^{-i \psi})m_{Y}(\psi)= L_{1}(e^{-i \psi}) L_{2 \times 12}(e^{-i \psi})m_{P}(\psi).
\label{bx_bp1}
\end{equation}
and the Assumption \ref{fourier_mean} holds.
\\In the general case, when  $p \in \mathbb{N}$ we use natural operator $L_{p}(B)=(1-B)^p$. The resulting time series $\ciag{X_{t}}$ can be represented by the following transformation of $Y_t$:
$$X_{t}=(1-B)^p Y_{t},$$
and hence, the expectation of $X_t$ takes the form:
\begin{equation}
E(X_t)=\mu_{X}(t)= p! \beta_{p} + \sum\limits_{\psi \in \Psi_{X}}m_{X}(\psi)e^{i \psi t},
\label{fouriera_produkcja13}
\end{equation}
where, according to the Theorem 8.1:
\begin{equation}
m_{X}(\psi) = L_{p}(e^{-i \psi})m_{Y}(\psi)= L_{p}(e^{-i \psi})
L_{2 \times 12}(e^{-i \psi})m_{P}(\psi).
\label{bx_bp2}
\end{equation}
By estimation $|L_{p}(e^{-i \psi})L_{2 \times 12}(e^{-i \psi})|>0$, which is true for any $p \in \mathbb{N}$ and $\psi \in (0,0.35)$, we have:
\begin{equation}
\Psi_X \cap (0,0.35)= \Psi_{P,1}. \label{for1}
\end{equation}
Therefore the problem of parameter identification and estimation in the set
$\Psi_{P,1}$ reduce to the problem of parameter identification and estimation in
the set $\Psi_{X} \cap (0,0.35)$.
\item[\textbf{Step 3 - }] \textbf{parameter identification and estimation.} The formula (\ref{for1}) is crucial in the algorithm of parameter identification and estimation in the set $\Psi_{P,1}$. Initially, in Step 3 we formulate the additional assumption that the autocovariace function of time series $\ciag{X_{t}}$ exists and it is almost periodic function. Notice, that the weaker     assumption concerning periodic structure of autocovariance function appears in the literature concerning analysis of economic time series; see for example: \cite{parzen}, \cite{osborn}, \cite{franses_96}, \cite{franses_96_book}, \cite{franses_97}, \cite{franses_05}. We use statistics  $\tilde \Pi_{n}(\{\psi\})=\sqrt{n}| \hat r_{n}(\psi)|$ and corresponding critical value $\tilde c_{n,b}(0.99\%)$ for the series generated from the previous steps of the algorithm.
    The test statistics $\tilde \Pi_{n}(\{\psi\})$ can be interpret as a value of test statistics $\Pi_{n}(\{\psi\})=\sqrt{n}| \hat m_{n}(\psi)|$ based on the sample $\{X_{1}-\overline{\textbf{X}}_{n},X_{2}-\overline{\textbf{X}}_{n},
\ldots,X_{n}-\overline{\textbf{X}}_{n} \}$, where $\overline{\textbf{X}}_{n}$ is the sample mean for the path $\{X_{1},X_{2},\ldots,X_{n}\}$. The critical value is calculated according to the formula that utilises subsampling procedure presented in \cite{Pol}:
$$\tilde g_{n,b}^{\{\psi\}}(1-\alpha)=\inf\{x: \tilde
G_{n,b}^{\{\psi\}}(x) \geq 1-\alpha \},$$ where $$ \tilde
G_{n,b}^{\{\psi\}}(x)=\frac{1}{n-b+1}\sum_{t=1}^{n-b+1} \mathbf{1}\{
\sqrt{b}|\hat r^{t-1,b}_{n}(\psi)| \leq x \}.$$
We fix $b=2.5\sqrt{n}$ and we calculate test statistics and corresponding critical value for $\psi$ from the discrete set of frequencies on the interval $(0,0.35)$. If the value of test statistics is greater than the critical value on some subinterval $I \subset (0,0.35)$ we take this subinterval as the interval containing some elements of the set $\Psi_{P,1}$. Next, we estimate the frequency connected with subinterval $I$ using (\ref{eqww}). By plug in technique we estimate amplitude related to each identified frequency in almost periodic part of the mean function of the process $\ciag{P_{t}}$.
\end{itemize}
In our algorithm step 3 is fundamental in procedure of extraction business cycle component from the observed time series. Its main advantage is, that the frequencies describing the cyclical dynamics of economic activity are subject to formal statistical inference. This clearly distinguishes our approach from many other procedures presented in the literature, where the lack of statistical uncertainty in the procedure is very common and forces \emph{ad-hoc} approach; see for example polemics concerning detrending in \cite{Canova_98} and \cite{Burnside_98}.
\\However, it is very important, that the procedure yields only statistically significant frequencies, and extraction of the business cycle is subject to additional filtering. In the empirical part of the paper we use the Hodrick-Prescott (HP) filter (see \cite{hodrick_97}), with smoothness parameter $\lambda$. According to \cite{gomez99}, \cite{Gomez_01}, \cite{Maravall_01} parameter $\lambda$ can be described as the argument of frequency $\psi_0$:
\begin{equation}
\lambda=\frac{1}{4(1-\cos(\psi_0))^2},
\label{gomez}
\end{equation}
where $2\pi/\psi_0$ can be interpret as a length of the cycle. Hence, on the basis of our procedure, it is possible to choose appropriate parameter $\lambda$ of the HP  filter, restricting spectrum only to significant parameters in the set $\Psi_{P,1}$. Alternatively it is possible to apply any filter in cycle extraction. We choose the simplest HP filter for illustrative purposes.
\section{Empirical illustration}
In this part of the paper we analyse cyclical behaviour of production sector in Poland. In particular we apply our model and three step procedure in order to characterise business cycles in industrial production index and some subsector indices.
\\Figure \ref{masa0} (a) presents time series of industrial production index\footnote{Source: Eurostat.} in Poland from January 1995 to December 2009 (2005 year = 100\%). This index contains: mining and quarrying; manufacturing; electricity, gas, steam and air conditioning supply. In the first step we applied centered moving average filter \ma to eliminate strong seasonal effects. The results of filtering is plotted on the Figure \ref{masa0} (b). It is clear, that centered moving average filter removes the seasonal effects and also business fluctuations are clearly observable (see Figure \ref{masa0} (b)).
\\According to our algorithm, presented in previous sections, we present in Figure \ref{masa4} (a) first differences of the centered moving average filter applied for industrial production index. We see some evidence about the existence of cyclical behavior in time series under consideration. The amplitude of cycle does not seem to be constant over time. Also, the amplitude is smaller in period 1995-2001, while after year 2001 is characterized by greater variability. Therefore we use logarithm transformation for industrial production index to stabilize the amplitude. Figure \ref{masa4} (b) presents the first difference of centered moving average filter applied for logarithm of industrial production index. It is easy to see that the amplitude is more constant over time then before logarithm transformation.
\\Figure \ref{masa5a} presents plots of the values of the test statistics $\tilde
\Pi_{n}(\{\psi\})=\sqrt{n}| \hat r_{n}(\psi)|$ with corresponding critical value
$\tilde c_{n,b}(0.99\%)$. The test statistics exceeds the critical value in three
subsets on the interval $(0,0.35)$. Hence, taking care only of significant values
of test statistic on the Figure \ref{masa5a} and in zoom on Figure \ref{masa6}, we
assume that:
\begin{equation}
\Psi_{P,1} \cap (0,0.35)=\{\psi_{1}, \psi_{2}, \psi_{3}\}.
\end{equation}
The values of estimated frequencies from the set $\Psi_{P,1}$ were calculated according to (\ref{eqww}). These values with corresponding length of the cycle can be found in Table 1.
\\Estimated amplitude of the cycle with corresponding length $8.5$  and $3.4$ years equal
$0.13$ and $0.07$ respectively. This second amplitude dominates the estimated value of amplitude with
corresponding length $2$ years. We see, that $8.5$-year length of the cycle received data support. However we can not formally interpret such fluctuations as a long-term growth trend or business fluctuations. We should rather look at this fluctuations as a mixture of both long-term growth trend and business fluctuations. Consequently and unquestionable, the dataset support  fluctuations with corresponding length $3.4$ years as a basic characteristic of business cycle in industrial production in Poland. To confirm this statement we extract cyclical fluctuations from industrial production index (filtered by centered moving average filter \ma) with the use of HP filter condition to the values of parameter $\lambda$ fixed for $\lambda=5\,500$, $\lambda=12\,000$, $\lambda=32\,000$, $\lambda=55\,000$. The results are plotted on Figure \ref{prod_HP}. By restricting parameter $\lambda$ to values stated above, according to the formula \ref{gomez}, we extract fluctuations with the length not greater than $4.5$, $5.5$, $7$ and $8$ years respectively. Since our goal was to extract only business fluctuations without significant influence of long-term growth trend, we restrict filtering only to fluctuations with corresponding length shorter than $8$ years.
\\Analysing plots presented on Figure \ref{prod_HP} it is possible to confirm the presence of cycles in
industrial production in Poland with estimated length in the interval 3-4 years (during the period 1995 - 2009). In Table \ref{exp_rec} we determined the periods of recessions and expansions in industrial production. We interpret turning points as margins of this periods. In most cases the recession is shorter than expansion. Consequently, our analysis confirm results discussed in the literature, that business cycle in industrial production for Poland display asymmetric behaviour. Also, the business cycle \emph{troughs} are rather sharper than \emph{peaks}, which is also typical for business cycles; see \cite{Hicks_50}, \cite{milas_06}.
\\In the next step we provide a more detailed analysis based on a formal identification of business cycles in sectors and subsectors of industrial production in Poland. We use the same statistical tools as for the total industrial production index. We considered all categories identified for industrial production. The set of all modelled indices are presented in Table \ref{datasets}.
\\Figure \ref{prod-dane} presents plots of logarithms of all considered indices. Repeating the procedure, initially applied for the total index, we use centered moving average filter \ma to remove seasonal effect from the data sets (see Figure \ref{prod-2x12ma}). First differences are presented on the Figure \ref{prod-2x12ma-roznice}. It is clear, that majority of indices exemplify cyclical pattern, just like in the case of index of total production, but with rather differential amplitudes and length.
\\To identify frequencies in the unknown set  $\Psi_{P,1}$ we applied again the test statistic $\tilde \Pi_{n}(\{\psi\})=\sqrt{n}|\hat r_{n}(\psi)|$ and corresponding critical value $\tilde c_{n,b}^{\{\psi\}}(\alpha)$. The results are presented on the Figure \ref{prod-2x12ma-roznice-test1}, where we plotted estimated lengths of the cycles together with appropriate estimated amplitude. In different sectors and subsectors of industrial production the data provide evidence in favour of cycles with length in the interval 1.5-3 years. However, those cycles are characterized by much shorter estimated amplitude than the cycles with length in the interval 3-4 years. The cycles with estimated length in the interval 3-4 years were supported in predominant set of subindices. Only in the case of manufacture, food products and beverages (C10\_C11), manufacture of basic pharmaceutical products and pharmaceutical preparations (C21) and electricity, gas, steam and air conditioning supply (D).
\\We see, that the largest estimated amplitude characterize cycles with the corresponding length of more than 4 years. But only in a few cases the cycles with length 5-8 years were supported. It can be seen clearly on the Figure \ref{prod-2x12ma-roznice-test1-amplitudaall}, where the comparison of all identified cycles for all 32 indices is presented. In spite of the fact, that observed time series were subject to filtering with the use of centered moving average filter \ma, all investigated subindices provide data support in favour of the existence of cycles with length not greater than two years. However, as seen on Figure \ref{prod-2x12ma-roznice-test1-amplitudaall}, those short cycles are characterized by amplitudes with values located relatively close to zero, as compared with longer significant cycles. This makes such a short term periodic pattern not extremely important in describing cyclical behaviour of modelled time series. Additionally, all indices support cycles of length 3-4 years, with relatively greater value of corresponding amplitudes as compared to characteristics of short term fluctuations. Also, except manufacture of wearing apparel (14-th index) we see no data support for cycles with length between 4 and 7 years. Consequently, for all considered subindices, the set of statistically significant cycles is clearly divided in two separate parts. The first set is constituted by short term cycles with small amplitudes together with middle term fluctuations, attributed in most cases by stronger amplitudes. The second set consists of frequencies, describing long term cycle, namely with length not less than 7 years. Just like in case of the total production index, we tend not to interpret those long term fluctuations as important characteristic of business cycle for Polish economy. According to our results, just like for the total index, all considered subindices are characterized by existence of the long term trend.
\\Using HP filter we extract business cycles from all industrial production indexes.
Similar as for industrial production index - total we fix the parameter $\lambda$ as
$\lambda=5\,500$, $\lambda=12\,000$, $\lambda=32\,000$, $\lambda=55\,000$ (see Figure
\ref{prod-2x12ma-roznice-test1-hp}). The reasons why we chose those values of $\lambda$ parameter are the following. Firstly, we fix the same parameter to compare results with those obtained for industrial production index. Secondly, the length of the cycle that is greater than 8 years is not clearly constant over different subindices and therefore we can not interpret those fluctuations as business fluctuations. Finally,  we can notice that in the interval from 4 to 8 years there are only a few significant lengths of cycles and this should give rise to extract similar shape of business fluctuations for different values of parameter $\lambda$. Almost all extracted fluctuations reveal presence of cycles with length in the interval 3-4 years. Summing  up, the cycle with length in the range 3-4 years is typical and prevalent for cyclical fluctuations in industrial production in Poland.
\section{Concluding remarks}
In this paper a novel approach in business fluctuations analysis for one dimensional
economic processes is proposed. Using theory of almost periodically correlated time series and subsampling procedure we consider a formal approach to estimate the length of business cycles. The main advantage of our approach is, that the business cycle characteristics are treated in formal way, and are subject to statistical inference. This clearly distinguishes presented framework from many filtering-based approaches, broadly considered in empirical applications. We model business fluctuations by parameters of discrete spectra of time series, under assumption that amplitude of this fluctuations is constant over time. Taking in consideration estimated length of the cycles we extract business fluctuations by HP filter for parameter of smoothness chosen on the basis of formal procedure.
\\The main conclusion presented in empirical illustration is that, during period 1995-2009, we confirm (using statistical tools) the presence of 3-4 years length of business cycle in industrial production index in Poland. This result was obtained either on the basis of the total index and also analysing subindices. This result confirms analyses conducted so far on the basis of Polish macroeconomic time series; see \cite{gghp}, \cite{adpw}, \cite{skrzypczynski_08} and \cite{skrzypczynski_06}. 
\\All indices and subindices supported significance of short term and middle term fluctuations, attaching relatively small amplitudes for periodicity with length less than 2 years. Additionally, in all time series we detected existence of longer term cycle (7-8 years), interpreted in this paper as a trend or a mixture of both trend and business cycle fluctuations.
\section{Appendix}
\begin{assumption} Let $\ciag{X_{t}}$ be APC time series such that for any
 $x \in [0,2 \pi)$ there exists a constant $B(x)$ (which dependents only on $x$), such that we have estimation
\begin{equation}
\sum\limits_{\psi \in \Psi \setminus \{x \} } \left|m(\psi) \emph{cosec}\left({\frac{\psi - x}{2}}
\right) \right| < B(x) < \infty,
\label{zalozenie-bx}
\end{equation}
\label{fourier_mean}
\end{assumption}
\begin{theorem}
Let $\ciagn{X_{t}}$ be a time series for which the expectation function exists and it
is almost periodic function of the form $\mu_{X}(t)=E(X_t)=\sum\limits_{\psi \in \Psi_{X}}m_{X}(\psi)e^{i \psi t}.$ We
assume that for the set $\Psi$ and corresponding Fourier coefficients $m(\cdot)$ the
Assumption \ref{fourier_mean} holds. Let
$L(B)=\sum_{j=-p}^{q}a_j B^j$ be a linear filter, where $p,q \geq 0$,
$\{a_j\}_{j=-p}^{q}$ is a sequence of real numbers, and $B^j X_{t}=X_{t-j}$
for any  $j \in \mathbb{Z}$. Then
$$E(Y_t)=\mu_{Y}(t)=\sum\limits_{\psi \in \Psi_{Y}}m_{Y}(\psi)e^{i \psi t},$$
where $\Psi_Y=\Psi_X$ and $m_{Y}(\psi)=L(e^{-i \psi})m_{X}(\psi)$. Additionally, assumption \ref{fourier_mean} holds for the set
$\Psi_{Y}$ and corresponding coefficients $m_{Y}(\cdot)$. \label{tw-filtracja}
\end{theorem}
\begin{proof}
Notice that
\begin{equation}
\begin{split}
E(Y_t)& = E\left(\sum_{j=p_1}^{p_2}a_j B^j X_t  \right)= E\left(\sum_{j=p_1}^{p_2}a_j X_{t-j}  \right)\\
& = \sum_{j=p_1}^{p_2}a_j \sum\limits_{\psi \in \Psi_{X}}m_{X}(\psi)e^{i \psi (t-j)} =
 \sum\limits_{\psi \in \Psi_{X}}m_{X}(\psi) \sum_{j=p_1}^{p_2} a_j e^{-i \psi j}
 e^{i \psi t} =  \sum\limits_{\psi \in \Psi_{X}}m_{X}(\psi) L( e^{-i \psi})
 e^{i \psi t}.
\end{split}
\end{equation}
By estimation $|m_{Y}(\psi)|\leq |m_{X}(\psi)| \sum\limits_{j=p_1}^{p_2}|a_j|$ we
conclude that condition 1.1 from \cite{Lenart_11_jtsa} holds for the set
$\Psi_{Y}$ and corresponding Fourier coefficients $m_{Y}(\cdot)$.
\end{proof}

\begin{theorem}
Take any $\psi \in (0,2 \pi)$. Let  the assumptions of Theorem 2.2 in \cite{Lenart_11_jtsa} hold.  Then
\begin{itemize}
\item[\emph{(i)}]$\tilde L_{n,b}^{\{\psi\}}(x) \stackrel{p}{\rightarrow}
    J^{\{\psi\}}(x)$, for any $x \in \mathbb{R}$,
\item[\emph{(ii)}]$\sup_{x \in \mathbb{R}}|\tilde L_{n,b}^{\{\psi\}}(x) -
    J^{\{\psi\}}(x)|\stackrel{p}{\longrightarrow}0$,
\item[\emph{(iii)}] subsampling confidence intervals for the parameter $|m(\psi)|$ are asymptotically consistent, which means that
\begin{equation}
P\left(\sqrt{n} \left(|\hat r_{n}(\psi)|-|m(\psi)|\right)\leq \tilde c_{n,b}^{\{\psi\}}(1-\alpha)\right)\longrightarrow
1-\alpha, \;\;\;
\label{con-sub1-beta}
\end{equation}
\end{itemize}
where $b=b(n) \rightarrow \infty$ and $b/n \to 0$.
\label{subsamplingp-mean}
\end{theorem}

\begin{proof}[\textbf{Proof of the Theorem \ref{subsamplingp-mean}.}] I this proof we use the same steps as in Theorem 4.2.1, page 103
in \cite{Pol}. Let $q=n-b+1$, $\tau_{n}=\sqrt{n}$ and
$$ U_{n}(x)=\frac{1}{q}\sum\limits_{t=1}^{q}\mathbf{1}\{\tau_{b}(|\hat m_{n}^{t-1,b}(\psi)|-|m(\psi)|)\leq
x\}.$$ Notice, that
$$\tilde L_{n,b}^{\{\psi\}}(x)=\frac{1}{q}\sum\limits_{t=1}^{q}
\mathbf{1}\{\tau_{b}[|\hat m_{n}^{t-1,b}(\psi)|-|m(\psi)|]
+\tau_{b}[(|m(\psi)|-|\hat r_{n}(\psi)|)+(|\hat r_{n}^{t-1,b}(\psi)|-|\hat m_{n}^{t-1,b}(\psi)|)]\leq
x\}.$$ We need the following lemma.
\begin{lemma}
For any real $x$ and $\epsilon>0$ we have estimation
\begin{equation}
 U_{n}(x-\epsilon)\mathbf{1}\{ E_{n}\}\leq \tilde L_{n,b}^{\{\psi\}}(x)\mathbf{1}\{ E_{n}\}\leq
 U_{n}(x+\epsilon),
\label{xx}
\end{equation}
where $ E_{n}=\{\tau_{b}\max\limits_{1 \leq t \leq q }|(|m(\psi)|-|\hat
r_{n}(\psi)|)+(|\hat r_{n}^{t-1,b}(\psi)|-|\hat m_{n}^{t-1,b}(\psi)|)|\leq
\epsilon\}$.
\label{lema}
\end{lemma}
\begin{proof} Let consider two cases:
\begin{itemize}
\item[$1^{o}$]$\mathbf{1}\{E_{n}\}=0$, inequality (\ref{xx}) holds
\item[$2^{o}$]$\mathbf{1}\{E_{n}\}=1$, then

$$\tau_{b}\max\limits_{1 \leq t \leq q
    }|(|m(\psi)|-|\hat r_{n}(\psi)|)+(|\hat r_{n}^{t-1,b}(\psi)|-|\hat
    m_{n}^{t-1,b}(\psi)|)| \leq \epsilon,$$ which means that for any  $1 \leq
    t \leq q$
    $$\eta_{n}^{t-1,b}(\psi):=\tau_{b}[(|m(\psi)|-|\hat
     r_{n}(\psi)|)+(|\hat r_{n}^{t-1,b}(\psi)|-|\hat
     m_{n}^{t-1,b}(\psi)|)]\in[-\epsilon,\epsilon].$$ Using next inequality $x-\eta_{n}^{t-1,b}(\psi) \geq x-\epsilon$, which is true for any $1 \leq t \leq q$ we get
\begin{equation}
\begin{split}
{ } & \mathbf{1}\{\tau_{b}[|\hat m_{n}^{t-1,b}(\psi)|-|m(\psi)|]
+\tau_{b}[(|m(\psi)|-|\hat r_{n}(\psi)|)+(|\hat r_{n}^{t-1,b}(\psi)|-|\hat m_{n}^{t-1,b}(\psi)|)]\leq x\}\\
 & = \mathbf{1}\{\tau_{b}[|\hat m_{n}^{t-1,b}(\psi)|-|m(\psi)|]\leq x-\eta_{n}^{t-1,b}(\psi)\}
\geq   \mathbf{1}\{\tau_{b}[|\hat m_{n}^{t-1,b}(\psi)|-|m(\psi)|]\leq x-\epsilon\}.
\end{split}
\label{x1}
\end{equation}
Analogically, using inequality  $x- \eta_{n}^{t-1,b}(\psi) \leq
x+\epsilon$ we get
\begin{equation}
\begin{split}
{ } & \mathbf{1}\{\tau_{b}[|\hat m_{n}^{t-1,b}(\psi)|-|m(\psi)|]
+\tau_{b}[(|m(\psi)|-|\hat r_{n}(\psi)|)+(|\hat r_{n}^{t-1,b}(\psi)|-|\hat m_{n}^{t-1,b}(\psi)|)]\leq x\}\\
 & = \mathbf{1}\{\tau_{b}[|\hat m_{n}^{t-1,b}(\psi)|-|m(\psi)|]\leq x-\eta_{n}^{t-1,b}(\psi)\}
\leq   \mathbf{1}\{\tau_{b}[|\hat m_{n}^{t-1,b}(\psi)|-|m(\psi)|]\leq x+\epsilon\}.
\end{split}
\label{x2}
\end{equation}
Summing inequality (\ref{x1}), (\ref{x2}) dla $t=1,2,\ldots,q$
we get (\ref{xx}). This completes the proof of lemma.
\end{itemize}
\end{proof}
In next step we show that $P(E_{n}) \to 0$. Using inequality $||z_1|-|z_2|| \leq |z_1-z_2|$
(which is true for any complex numbers $z_1,z_2$) and inequality $|\sum_{j=p}^{q} c_j e^{i j x}| \leq c_p |\text{cosec}(x/2)|,$ (which is true for any
$x \not \equiv 0 \text{ modulo } 2 \pi$ and real numbers $c_p \geq c_{p+1} \geq \ldots \geq c_q$
) we have
\begin{equation}
\begin{split}
\max\limits_{1 \leq t \leq q } |\eta_{n}^{t-1,b}(\psi)|
    & \leq \max\limits_{1 \leq t \leq q } \tau_{b}||m(\psi)|-|\hat r_{n}(\psi)||
    + \max\limits_{1 \leq t \leq q } \tau_{b} ||\hat r_{n}^{t-1,b}(\psi)|-|\hat
    m_{n}^{t-1,b}(\psi)||\\
    & \leq \tau_{b}||m(\psi)|-|\hat r_{n}(\psi)|| + \max\limits_{1 \leq t \leq q }
    \tau_{b} |\hat r_{n}^{t-1,b}(\psi)-\hat m_{n}^{t-1,b}(\psi)|\\
    & \leq \tau_{b}|m(\psi)-\hat r_{n}(\psi)| + \frac{\tau_{b}}{b} |\overline{\textbf{X}}_{n}||\cosec(\psi/2)|\\
    & \leq \tau_{b}|m(\psi)-\hat m_{n}(\psi)| + \tau_{b}|\hat m_{n}(\psi) -\hat r_{n}(\psi)| + \frac{\tau_{b}}{b} |\overline{\textbf{X}}_{n}||\cosec(\psi/2)|\\
    & \leq \tau_{b}|m(\psi)-\hat m_{n}(\psi)| +\frac{\tau_{b}}{n} |\overline{\textbf{X}}_{n}||\cosec(\psi/2)| + \frac{\tau_{b}}{b} |\overline{\textbf{X}}_{n}||\cosec(\psi/2)|\\
\end{split}
\end{equation}
By convergence  $\tau_{b}|m(\psi)-\hat m_{n}(\psi)|\stackrel{p}{\rightarrow }0$ and $\frac{\tau_{b}}{b}
|\overline{\textbf{X}}_{n}| \stackrel{p}{\rightarrow }0$ we get
$$\max\limits_{1 \leq t \leq q } |\eta_{n}^{t-1,b}(\psi)|\stackrel{p}{\rightarrow }0,$$
which means that  $P(E_{n}) \to 1$. Using next Slutsky's Lemma and Theorem 2.1 in \cite{Lenart_11_jtsa}
we have $\sqrt{n}(|r_{n}^{t-1,b}(\psi)|-|m(\psi)|)
\stackrel{d}{\rightarrow }J^{\{\psi\}}$.
To finish the proof it is sufficient to follows next steps in Theorem 4.2.1, page 103 in \cite{Pol}, therefore we omit them.
\end{proof}

\bibliographystyle{plainnat} 
\bibliography{biblioMateusz}
\clearpage
\begin{figure}
\begin{center}
\subfigure[\scriptsize{ }]
{\includegraphics[width=7.5 cm]{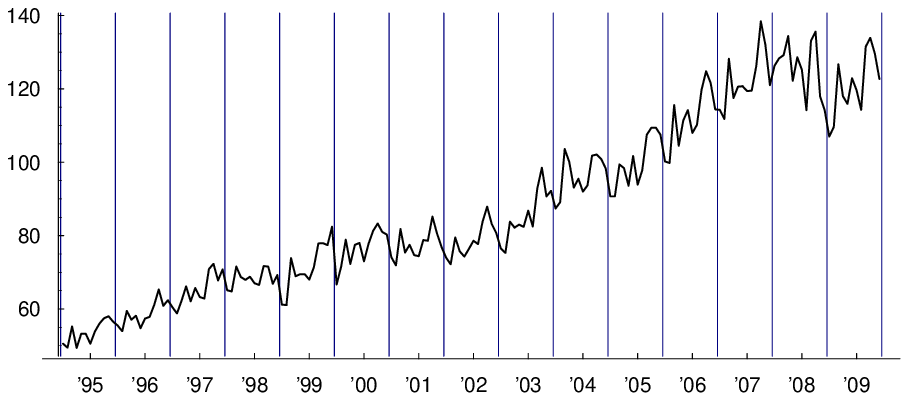}}
\subfigure[\scriptsize{ }]
{\includegraphics[width=7.5 cm]{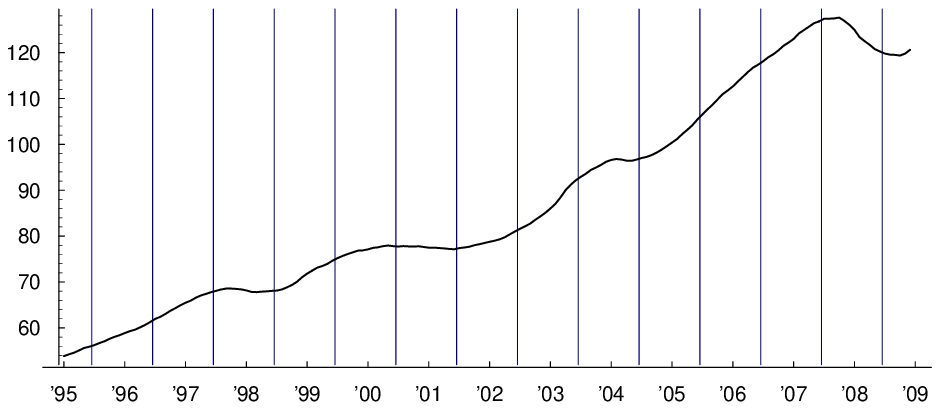}}
\end{center}
\caption{\small{(a)  Industrial production index in Poland (2005 year = 100\%) from January 1995 to December 2009; (b)
Realization of centered moving average filter \ma applied for industrial production index in Poland.}}
\label{masa0}
\end{figure}
\begin{figure}
\begin{center}
\subfigure[\scriptsize{ }]
{\includegraphics[width=7.5 cm]{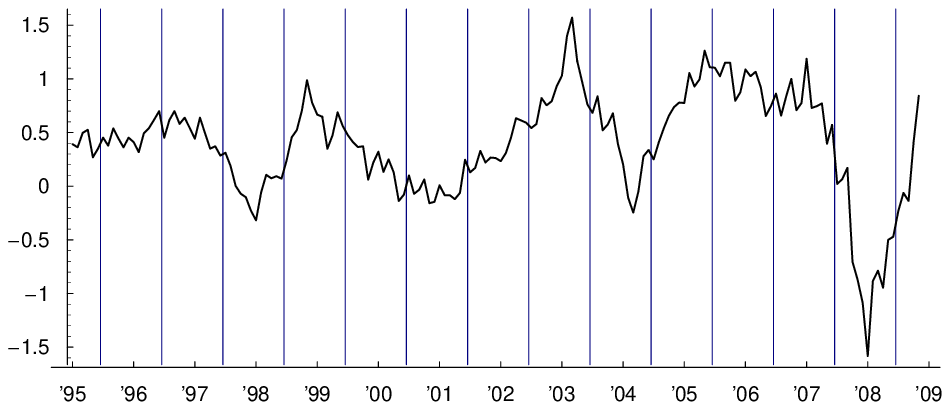}}
\subfigure[\scriptsize{ }]
{\includegraphics[width=7.5 cm]{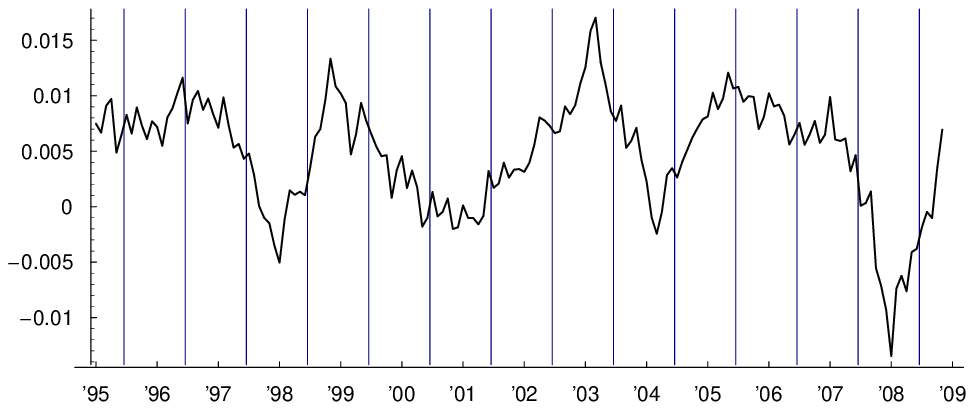}}
\end{center}
\caption{\small{(a)  First difference of centered moving
average filter \ma applied for industrial production index; (b)
First difference of centered moving
average filter \ma applied for logarithm of industrial production index.}}
\label{masa4}
\end{figure}
\begin{figure}[H]
\begin{center}
{\includegraphics[width=16 cm]{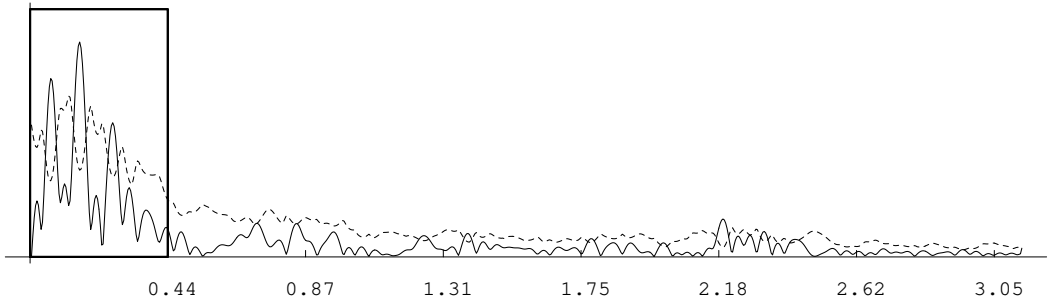}}
\end{center}
\caption{\small{Frequency identification using statistics
$\tilde \Pi_{n}(\{\psi\})=\sqrt{n}|\hat r_{n}(\psi)|$ and corresponding critical value
$\tilde c_{n,b}(0.99\%)$ for the realization of time series
$\ciag{X_{t}}$: continuous line - the value of test statistics  $\tilde \Pi_{n}(\{\psi\})=\sqrt{n}|\hat r_{n}(\psi)|$ for  $\psi$ from the set
$\{ (k-1) \pi/720: \: k=1,2,\ldots,720 \}$;  dashed line - critical value
$\tilde c_{n,b}^{\{\psi\}}(99\%)$  for  $\psi$ from the set
$\{ (k-1) \pi/720: \: k=1,2,\ldots,720 \}$.}} \label{masa5a}
\end{figure}
\begin{figure}
\begin{center}
\subfigure[\scriptsize{ }]
{\includegraphics[height=3.0 cm, width=7 cm]{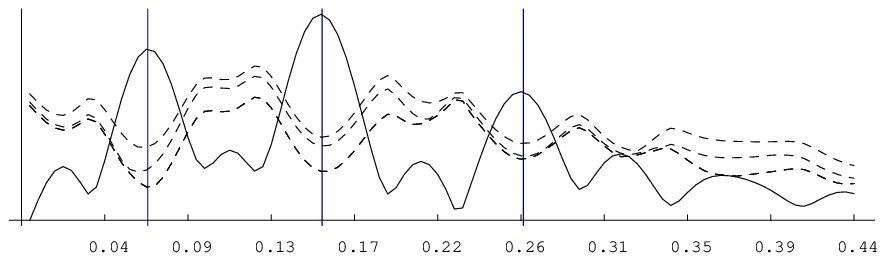}
\label{masa6a}}
\subfigure[\scriptsize{ }]
{\includegraphics[height=4.0 cm, width=7 cm]{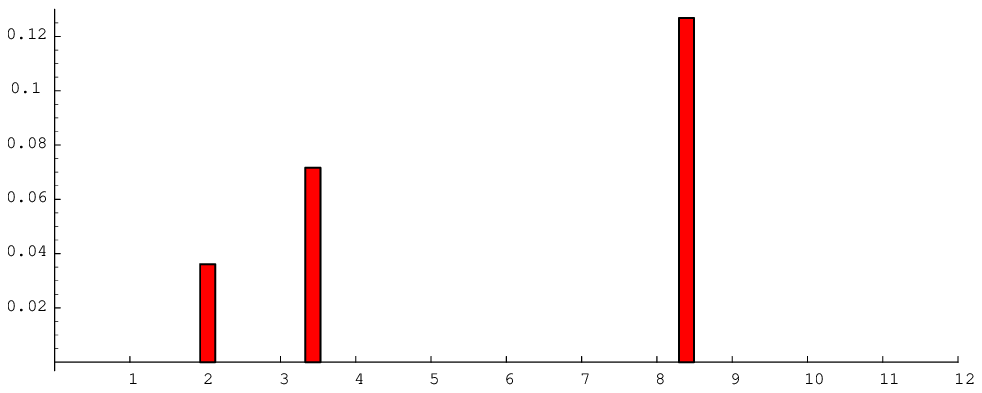}}
\end{center}
\caption{\small{Frequency identification (in the set $\Psi_{P,1}$)
and corresponding amplitude estimation: (a)  the value of test statistics  $\tilde \Pi_{n}(\{\psi\})=\sqrt{n}|\hat
r_{n}(\psi)|$  (continuous line) and critical value $\tilde c_{n,b}^{\{\psi\}}(\alpha)$
(dashed line) for  $\alpha \in \{ 92\%, 95\%, 99\% \}$ and $\psi$ from the set $\{
(k-1) \pi/720: \: k=1,2,\ldots,100 \}$; (b) estimated amplitude corresponding to
estimated frequencies from the set $\Psi_{P,1}$:  X - estimated length of the cycle,
Y - estimated amplitude.}} \label{masa6}
\end{figure}
\begin{table}
\begin{center}
{\small
\begin{tabular}{|c|c|c|c|}
  \hline
  & & & \\
  The value of frequency  & $\hat \psi_{n,1}=0.062$ & $\hat \psi_{n,2}=0.153$ & $\hat \psi_{n,3}=0.258$  \\
  estimator & & & \\
    & & & \\
  \hline
    & & & \\
  {Corresponding length} &  & &     \\
   {of the cycle}  & $8.5$  & $3.4$   &$2$    \\
     {(in years)} &   &    &    \\
       & & & \\
\hline
\end{tabular}
} \caption{\small{Estimated frequencies with corresponding length of the cycle for
industry production index in Poland.}}
\end{center}
\label{tab51}
\end{table}
\begin{figure}
\begin{center}
{\includegraphics[height=5 cm, width=15 cm]
{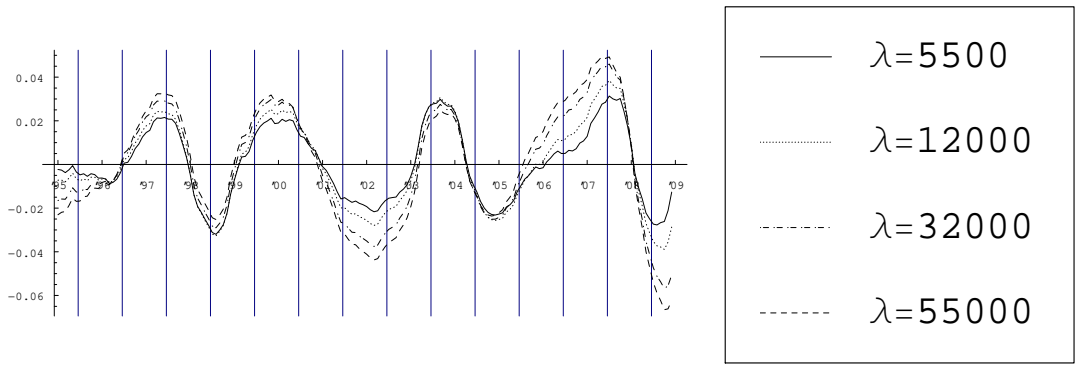}}
\end{center}
\caption{\small{Business cycle in industrial production after
logarithm and application of centered moving average filter extracted by HP filter for
 $\lambda=5\,500$ (continuous line) $\lambda=12\,000$ (dotted line)
$\lambda=32\,000$ (doted and dashed line) $\lambda=55\,000$ (dashed line).}}
\label{prod_HP}
\end{figure}

\begin{table}
\begin{center}
\scriptsize{
\begin{tabular}{|c|c|c|c|c|c|}
  \hline
 &       &   &   &   &   \\
  Expansion & ... - Dec. 97       & Feb. 99 - May 00 & Sept. 02 - Mar. 04 &
  Jun. 05 - Jan. 08 &  Apr. 09 - ...\\
   &       &   &   &   &   \\
\hline
 &       &   &   &   &   \\
  Recession   & Dec. 97 - Feb. 99 & May 00 - Sept. 02 & Mar. 04 - Jun. 05 & Jan. 08 - Apr. 09 &   \\
   &       &   &   &   &   \\
  \hline
\end{tabular}
} \end{center} \caption{\small{Expansions and recessions in industrial production
index in Poland in the period July 1995 - June 2009.}}
\label{exp_rec}
\end{table}
\begin{table}
\begin{center}
\scriptsize{
\begin{tabular}{l}
\hline
B-D\_F -  Mining and quarrying; manufacturing; electricity, gas, steam and air conditioning supply; construction \\
MIG\_ING\_CAG - MIG  Intermediate and capital goods\\
MIG\_ING - MIG -  MIG - Intermediate goods\\
MIG\_CAG - Capital goods\\
MIG\_DCOG - MIG - Durable consumer goods\\
MIG\_NDCOG - MIG - Non-durable consumer goods\\
B - Mining and quarrying\\
C - Manufacturing\\
C10-C12 - Manufacture of food products; beverages and tobacco products\\
C10\_C11 - Manufacture of food products and beverages\\
C10 -  Manufacture of food products\\
C11 -  Manufacture of beverages\\
C12 -  Manufacture of tobacco products\\
C13\_C14 - Manufacture of textiles and wearing apparel\\
C15 - Manufacture of leather and related products\\
C16 - Manufacture of wood and of products of wood and cork, except\\
    furniture; manufacture of articles of straw and plaiting materials\\
C17 - Manufacture of paper and paper products\\
C18 - Printing and reproduction of recorded media\\
C19 - Manufacture of coke and refined petroleum products\\
C20 - Manufacture of chemicals and chemical products\\
C21 - Manufacture of basic pharmaceutical products and pharmaceutical
    preparations\\
C22 - Manufacture of rubber and plastic products\\
C23 - Manufacture of other non-metallic mineral products\\
C24 - Manufacture of basic metals\\
C25 - Manufacture of fabricated metal products, except machinery and
    equipment\\
C26 - Manufacture of computer, electronic and optical products\\
C27 - Manufacture of electrical equipment\\
C28 - Manufacture of machinery and equipment n.e.c.\\
C29 - Manufacture of motor vehicles, trailers and semi-trailers\\
C29\_C30 - Manufacture of motor vehicles, trailers, semi-trailers and
    of other transport equipment\\
C31 - Manufacture of furniture; other manufacturing\\
D - Electricity, gas, steam and air conditioning supply\\
\hline
\end{tabular}
}\end{center}
\caption{\small{Categorised indices describing changes in economic activity in sectors and subsectors of industrial
production in Poland}}
\label{datasets}
\end{table}
\begin{figure}
\begin{center}
\vspace{-1 cm}
{\includegraphics[height=21 cm, width=17 cm]
{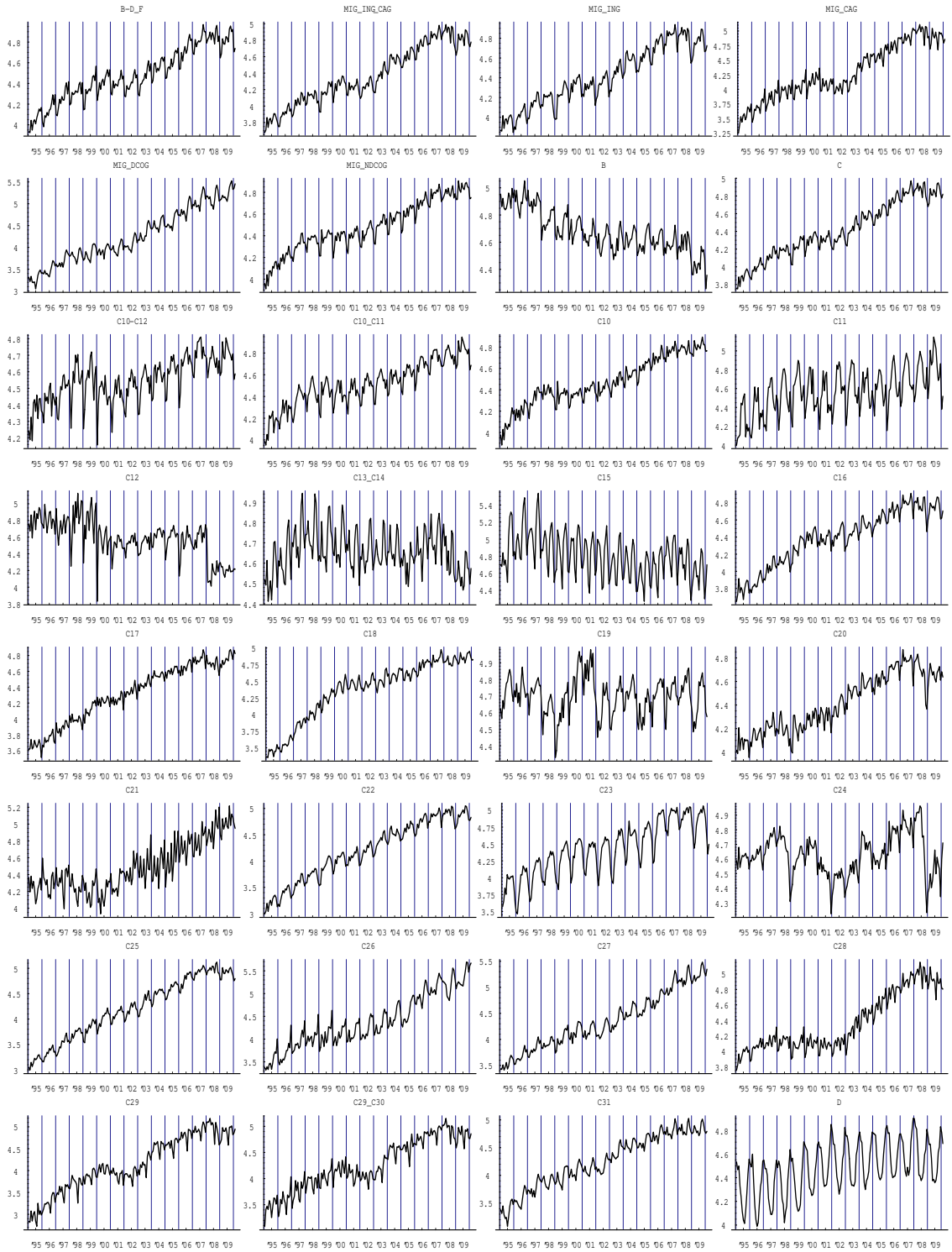}}
\end{center}
\vspace{-0.5 cm}
\caption{\small{Logarithm of industrial production
indices in Poland (2005 rok = 100\%) in sectors and subsectors from January 1995 to February 2010.}}
\label{prod-dane}
\end{figure}

\begin{figure}
\begin{center}
\vspace{-1 cm}
{\includegraphics[height=21 cm, width=17 cm]
{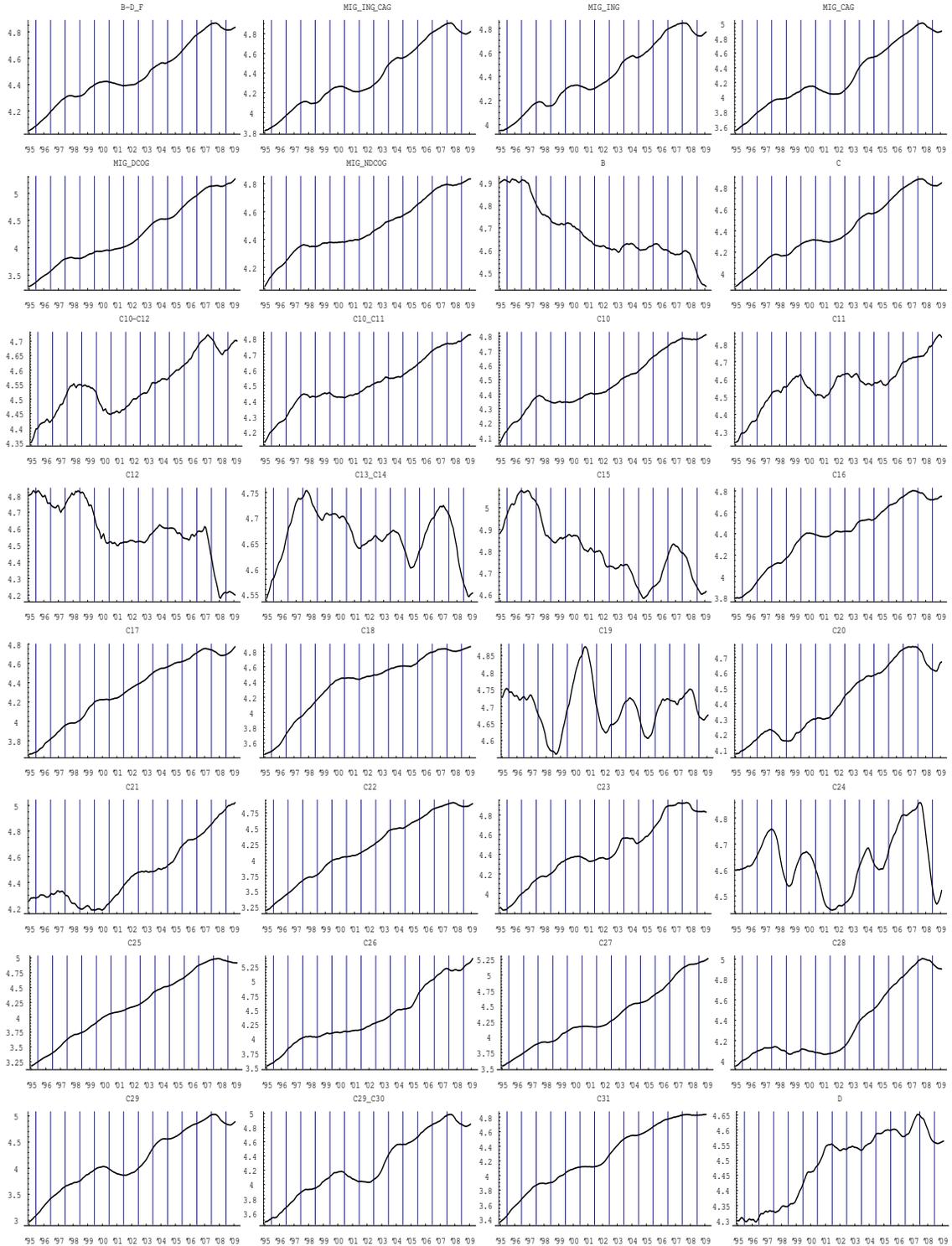}}
\end{center}
\vspace{-0.5 cm}
\caption{\small{Realizations of centered moving average filter \ma applied for logarithm of industrial production
indexes in Poland in sectors and subsectors.}}
\label{prod-2x12ma}
\end{figure}

\begin{figure}
\begin{center}
\vspace{-1 cm}
{\includegraphics[height=21 cm, width=17 cm]
{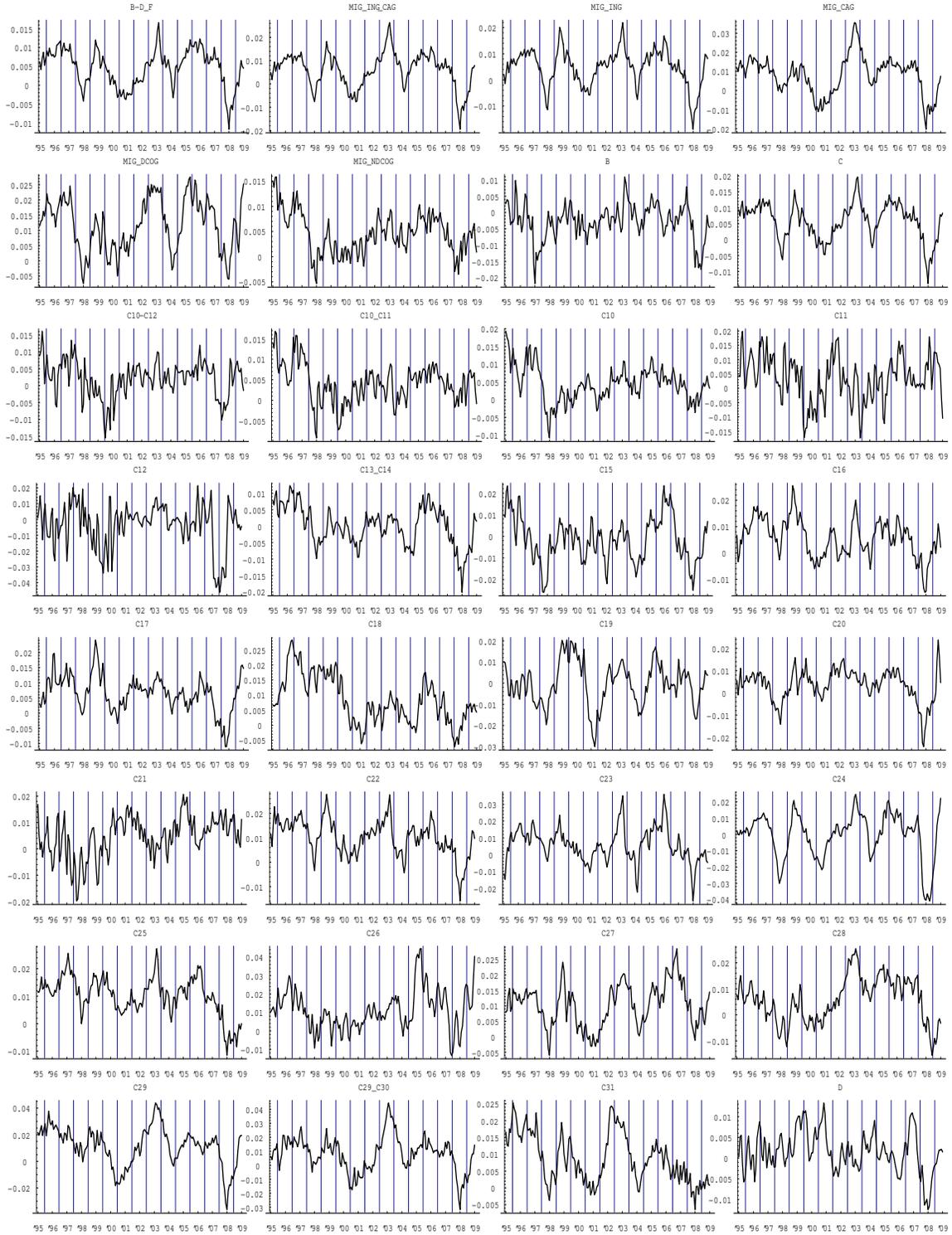}}
\end{center}
\vspace{-0.5 cm}
\caption{\small{First difference for realization of centered moving
average filter \ma applied for logarithm of industrial production
indexes in sectors and subsectors.}}
\label{prod-2x12ma-roznice}
\end{figure}

\begin{figure}
\begin{center}
{\includegraphics[height=19.5 cm, width=17 cm]
{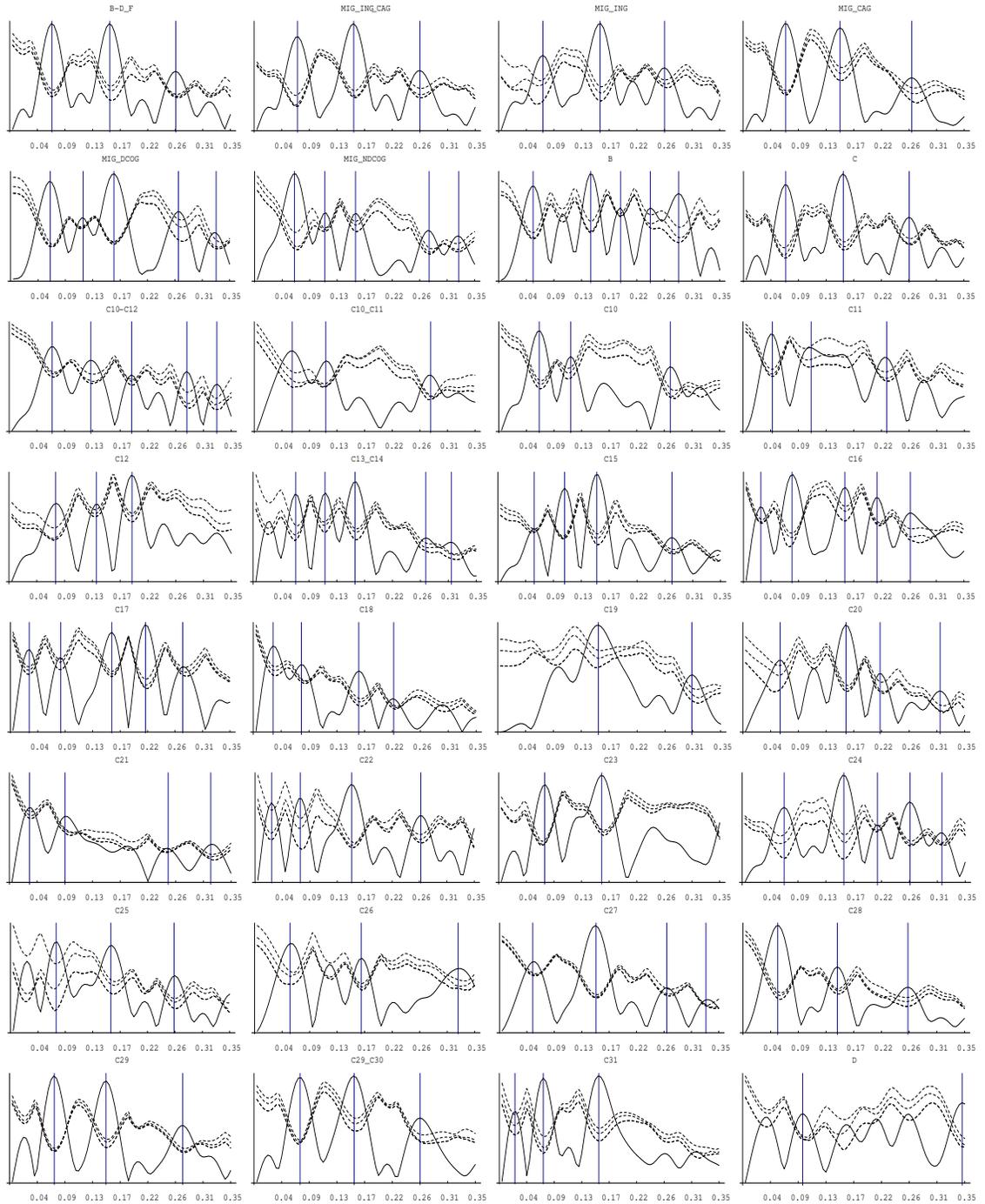}}
\end{center}
\caption{\small{Frequency identification (in the set $\Psi_{P,1}$): continuous line -  the value of test statistics  $\tilde \Pi_{n}(\{\psi\})=\sqrt{n}|\hat
r_{n}(\psi)|$, dashed line - critical value $\tilde c_{n,b}^{\{\psi\}}(\alpha)$
 for  $\alpha \in \{ 92\%, 95\%, 99\% \}$ and $\psi$ from the set $\{
(k-1) \pi/720: \: k=1,2,\ldots,100 \}$.}} \label{prod-2x12ma-roznice-test1}
\end{figure}

\begin{figure}
\begin{center}
\vspace{-1 cm}
{\includegraphics[height=19.5 cm, width=17 cm]
{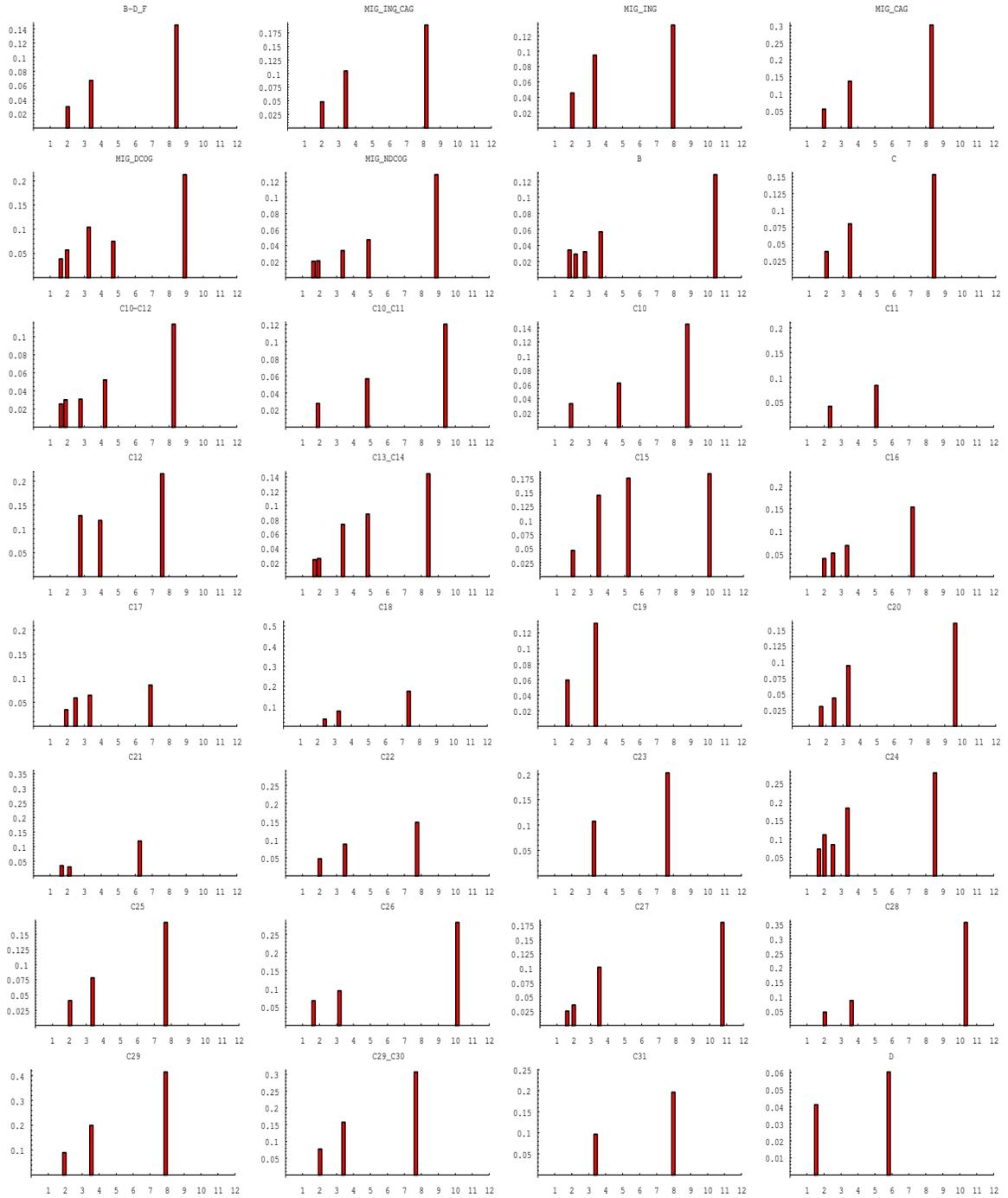}}
\end{center}
\vspace{-0.5 cm}
\caption{\small{Estimated amplitude and estimated
length of the cycles connected with identified frequencies in the set $\Psi_{P,1}$:  X -
estimated length of the cycle, Y - estimated amplitude.}}
\label{prod-2x12ma-roznice-test1-amplituda}
\end{figure}

\begin{landscape}
\begin{figure}
\begin{center}
{\includegraphics[height=10 cm, angle=0]
{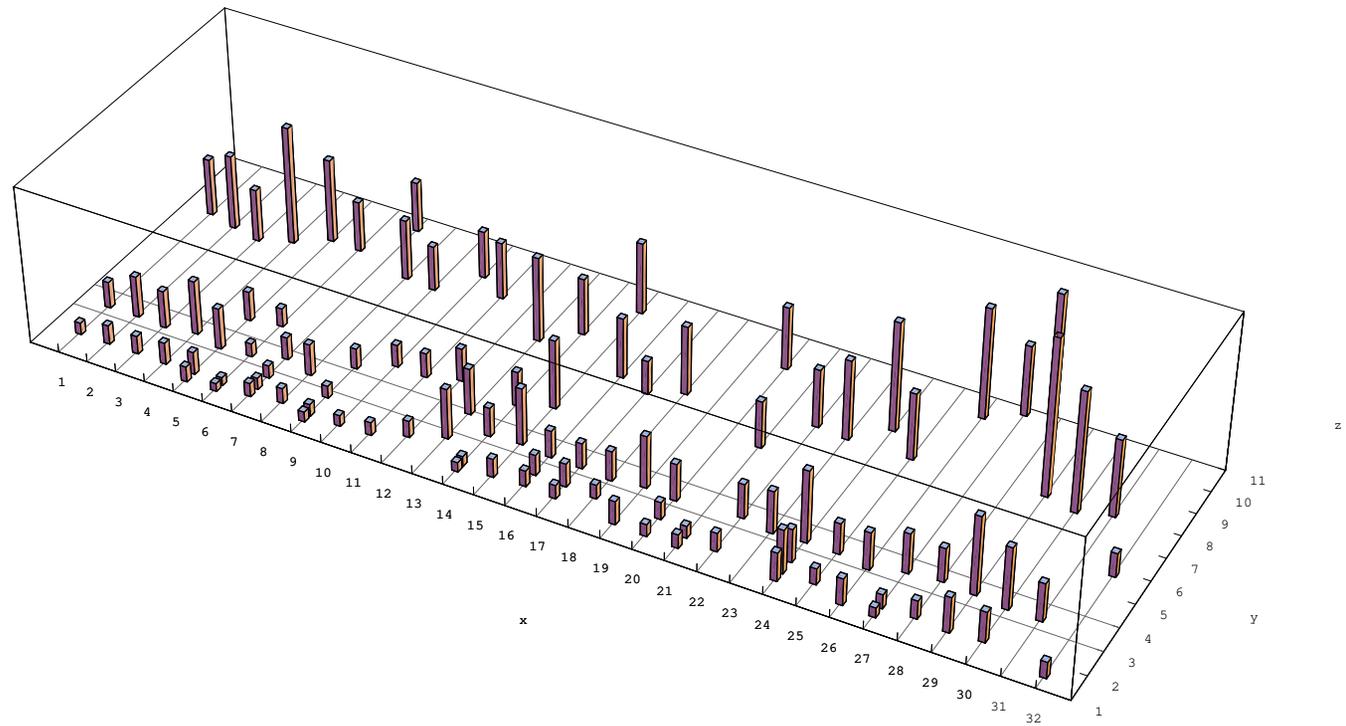}}
\end{center}
\caption{\small{Length of the cycles and amplitudes comparison:
X - the index number, Y - estimated length of the cycle, Z - estimated
amplitude.}} \label{prod-2x12ma-roznice-test1-amplitudaall}
\end{figure}
\end{landscape}

\begin{figure}
\begin{center}
\vspace{-1 cm}
{\includegraphics[height=20.5 cm, width=17 cm]
{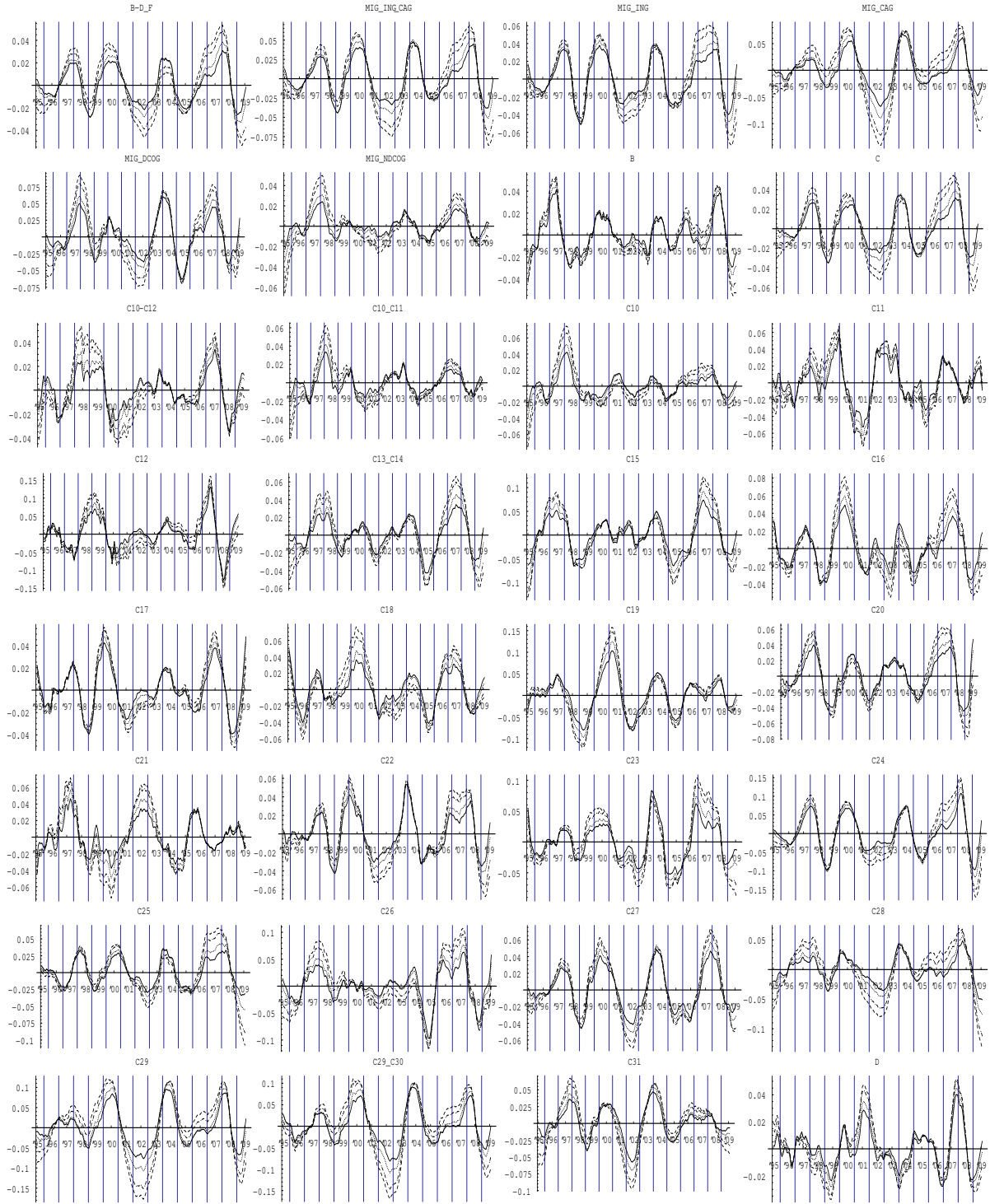}}
\end{center}
\vspace{-0.5 cm}
\caption{\small{Business cycle (extracted by HP filter) in sectors and subsectors of industrial production after
logarithm and application of centered moving average filter  for
 $\lambda=5\,500$ (continuous line) $\lambda=12\,000$ (dotted line)
$\lambda=32\,000$ (doted and dashed line) $\lambda=55\,000$ (dashed line).}}
\label{prod-2x12ma-roznice-test1-hp}
\end{figure}
\end{document}